\documentclass[sigconf,nonacm]{acmart}
\newtheorem{definition}{Definition}
\newtheorem{theorem}{Theorem}

\usepackage[ruled,vlined, linesnumbered]{algorithm2e}
\usepackage{algorithmic}
\usepackage{siunitx}
\usepackage{booktabs}
\usepackage{subfig}
\usepackage{multirow}

\begin{document}

\title{Some Optimization Solutions for Relief Distribution}

\author{Jhoirene  Clemente}
\email{jbclemente@up.edu.ph} 
\author{Jessie James Suarez}
\email{ jpsuarez@up.edu.ph}
\orcid{1234-5678-9012}

\affiliation{%
  \institution{Department of Computer Science\\
  University of the Philippines Diliman}
\country{Philippines}
  \postcode{1101}
}

\author{Olivia Demetria}
\email{olivia.demetria@maroonstudios.com}

\author {Perry Go}
\email{perry.go@maroonstudios.com.ph}

\affiliation{%
  \institution{Maroon Studios Inc.\\}
\country{Philippines}
  \postcode{1101}
}
\author{ Dylan Salcedo }
\email{dylan.salcedo@maroonstudios.com}

\affiliation{%
  \institution{Maroon Studios Inc.\\}
\country{Philippines}
  \postcode{1101}
}


\begin{abstract} Humanitarian logistics remain a challenging area of application for operations research. In relief distribution, the main goal is to deliver all the supplies to those that are in need in the fastest way possible. In this paper, we present different optimization solutions for relief distribution. We present a formalization of the three main problems in the humanitarian logistics aspect of relief distribution.  We identify the optimal location of the distribution centers.  We match the number of supplies to the number of demands for each distribution center based on the distribution of demands.  We provide the assignment of tasks to delivery fleet according to the location and the road network of the region. For each delivery truck, we provide an optimal sequence of visits to pre-assigned distribution centers. 

\end{abstract}

%\begin{CCSXML}
%<ccs2012>
%  <concept_id>10010520.10010553.10010562</concept_id>
%  <concept_desc>Computer systems organization~Embedded systems</concept_desc>
%  <concept_significance>500</concept_significance>
%  <concept_id>10010520.10010575.10010755</concept_id>
%  <concept_desc>Computer systems organization~Redundancy</concept_desc>
%  <concept_significance>300</concept_significance>
%  <concept_id>10010520.10010553.10010554</concept_id>
%  <concept_desc>Computer systems organization~Robotics</concept_desc>
%  <concept_significance>100</concept_significance>
%  <concept_id>10003033.10003083.10003095</concept_id>
%  <concept_desc>Networks~Network reliability</concept_desc>
%  <concept_significance>100</concept_significance>
%</ccs2012>
%\end{CCSXML}
%
%\ccsdesc[500]{Computer systems organization~Embedded systems}


\maketitle
\section{Introduction}
In the last 10 years, the United Nations  and the World Bank  estimates that, disasters, both natural and man-made, 
have killed more than 700 million people,  have affected and displaced more than 1.7 billion lives, and have wiped out 1.4 trillion US Dollars in assets and livelihood. Governments all over the world have been investing in constant disaster preparation and remediation efforts.  In the United States alone, Federal Emergency Management Agency (FEMA) gets roughly 13 billion US dollars in annual budget.  For developing countries such as the Philippines,  local and national leaders are mandated to allocate a minimum of $5\%$ of the annual budget to disaster risk reduction efforts. 

In this study, we present a computational solution to create an overall logistics plan for the distribution of relief goods using graph analysis and optimization applied to the road network of the region. The logistics plan includes the optimal number and location relief warehouses, the resource allocation based on the demand, as well as routes taken by the delivery trucks during deployment. As a case study, we used the road network and population distribution of the Province of Isabela.

%TODO https://pubfiles.pagasa.dost.gov.ph/

%TODO: Add statistics how the Province is a good case study
 %TODO: Add results here, observations and discussions
\section{Data Set}
%The province of Isabela consists of 37 municipalities with a total land area of $12,414.93\ km^2$, ranked as the $2$nd largest province in the country. The 37 municipalities include a total $1,055$ Barangays\footnote{A barangay is the smallest political unit in the Philippines, synonymous to a community}.  As of 2015, the total population of the province is  $1,593,566$.  

In this study, we made use of the road network extracted from the open street map data of the province. We also made use of the population distribution per municipality from the published data of the National Statistics Office. 

\subsection{Graph Representation of the Region}
Computationally, we model the road network using a complete edge-weighted graph $G =(V, E)$ with the set of vertices  $V$ including all the 37 municipalities of the province. For each vertex $i \in V$, we have $p_i\in \mathcal{Z}^+$ to denote the population of municipality $i$. For each municipality, we take into account the location by taking a representative point denoted by a a decimal pair $(x_i, y_i)$ from the global coordinate system. The $x_i$ and $y_i$ represents the latitude and longitude of the location respectively.  Initially, although this location can be arbitrary point inside a political region, this point can be an actual location of an evacuation center with a capacity of  $p_i$. For simplicity, we assume the ideal case that there is one evacuation center per municipality, and each vertex in the graph is a representation of the evacuation center.  However, the representation is not limited to this specific case and in fact it can also handle the case of multiple evacuation centers per municipality.  

An edge in a graph represents a path from one evacuation center to another. The edge weights are computed by getting the shortest distance between the two evacuation centers given the actual road network of the region.

The resulting complete edge weighted graph is the input to several computational problems which arise from creating the actual logistics plan. Let us discuss in detail the different information that is necessary to carry out a relief distribution effort. 

\section{Logistic Solutions}
In the area of Humanitarian Logistics (HL), disaster risk reduction activities can be divided into two stages: a {\it pre-disaster stage} and a {\it post-disaster stage}. In the pre-disaster stage, activities related to mitigation and preparation are involved. These includes the evacuation of people from disaster-stricken areas to safe places and planning the actual flow and storage of goods and materials from the point of origin to the point of consumption.  On the other hand, the post-disaster stage involves activities related to response and recovery 

In this study, we focus on the pre-disaster stage where we provide a logistics plan for the actual distribution of relief goods. In the logistic plan, we have the following  assumptions
\begin{itemize}
\item Residents in danger were already transferred to evacuation centers. 
\item If residents are safe within the municipality, i.e., no evacuation is necessary, every resident should be reachable by all the distributors within the municipality.
\item We do not take into account the limited capacity of delivery trucks.
\end{itemize}
In this study we focus on providing answers for creating a logistics plan that includes the following. 
\begin{enumerate}
\item The optimal number and location relief warehouses such that every evacuation center is reachable in the minimum amount of time
\item The resource allocation of relief goods to each of the chosen relief warehouses
\item Given an arbitrary number of relief trucks, we provide a division of assigned delivery locations oblivious of the political boundaries of the region
\item The tour taken by each delivery truck ensures the fastest turn around time
\end{enumerate}
Each of the identified information is discussed in each of the following  subsections.

\subsection{Facility Location} \label{facility_location}
As part of pre-disaster preparations, we seek to identify the location of the relief warehouses or distribution centers that can supply the demand of all the neighboring evacuation centers in the minimum amount of time. Our approach is to use the   facility location  to formally model the problem. 

The facility location problem is used to select optimal  location of shelters, distribution centers, warehouses, and medical centers subject to available input parameters, such as the number of affected population and location/capacity of candidate facilities.  Given that we have an unlimited amount of resources such as delivery trucks and man power, we can ideally deploy one truck for each evacuation center to ensure that supplies can reach the evacuations faster. However, this is often not the case during emergency operations, thus selecting an optimal number of drop points is necessary. In this paper, we solve the following  version of the facility location problem.
\begin{definition}[Relief Warehouse Selection Problem] \label{facloc}
Given a complete edge-weighted graph $G = (V,E,d)$ with a set of candidate relief warehouses  $V = \{1, 2, \ldots, n\}$, and $d: V x V  \rightarrow R^+$, where $d(i,j)$ is the shortest road distance between every pair of $i$ and $j$,  we seek to identify a subset of selected warehouses  $W \subseteq V$, such that 
\begin{enumerate}
\item for all $i \in V \setminus W $, node $i$ has at least one edge incident to $W$, and 
\item the total number of relief warehouses  selected $|W|$  and the cost of servicing the demand points, i.e., 
$$cost(W, G) = \sum_{i \in V \setminus W} \min_{\forall j \in W} d(i,j),$$ is minimum.
\end{enumerate}

\end{definition}
The first condition in Definition \ref{facloc} ensures that every candidate relief warehouse that is not in $W$ is connected to at least one vertex in $W$. Simply, every evacuation center is near a relief warehouse. Moreover, satisfying condition two in Definition \ref{facloc} ensures minimal resources in setting up a warehouse and minimal resources in servicing the remaining demand points. 

\begin{algorithm} \label{optima_facloc}
\begin{verbatim}
T  = MST(G)
W = MVC(T)
Return W
\end{verbatim}
\caption{Polynomial-time algorithm for Relief Warehouse Selection Problem}
\end{algorithm}

Let ALG be an algorithm for the relief warehouse selection problem. ALG  is composed of two stages. The first stage produces a minimum spanning tree  $T$  from $G$. The second stage produces a subset of vertices $W$ from the minimum vertex cover of $T$. 

\begin{theorem}
ALG produces an optimal solution for the relief warehouse selection problem in polynomial-time.
\end{theorem}

\begin{proof}
We need to show that the selected set of nodes $W$ from our solution satisfies the two conditions from Definition \ref{facloc} and that solution $W$ is optimal.
 
The resulting  minimum spanning tree $T$ from the first stage of ALG produces a connected subgraph of $G$.   The tree property ensures that there exists a simple path between every pair of vertices.  The selected warehouses is $W$ which is obtained by getting the minimum vertex cover of $T$.  Since,  the vertex cover ensures that every edge is incident to $W$, every vertex not in $W$ is incident to at least one vertex in $W$. Thus, satisfying condition 1 of Definition \ref{facloc}.  

To show that condition 2 is met,  we have the following proof.  By definition of the minimum spanning tree, $T$ is consists of the minimum weight edges  in $G$.  Since every vertex in $G$  is still connected to at least  one vertex in $W$ using the minimum spanning tree $T$,  then $cost(W,G) = \sum_{\forall i,j \in T} d(i,j),$,  which is minimum for all possible connected subgraph in $G$. 

Lastly,  ALG runs in polynomial-time because each stage runs in polynomial-time.  First,  computing for the minimum spanning tree has an $O(n\log{n}) $ solution and getting the minimum vertex cover in trees has a polynomial-time solution using the maximal matching as subroutine. 
\end{proof}
\subsection{Resource Allocation} \label{resource_allocation}
As the name implies, resource allocation seeks to answer the total number of supplies to deliver for each selected warehouse. We assume that the number of supplies from the supply port, i.e., the  source of all relief goods is not less than the required supply for the whole population.  Based on the minimum spanning tree obtained from the facility location algorithm,  we computed a resource allocation that takes into account the demand of each evacuation centers.

Our approach uses the solution of the selected relief warehouse problem. The computation of each resource allocation which follows an iterative procedure is shown below.

\begin{algorithm}
\begin{verbatim}
T = MST(G)
W = MVC(T)
for each i in  V\W:
        W_i =  the set of warehouse incident to i in T 
       for each j in W_i: 
                r_j = p_j + p_i/|W_i|
        r_i = 0
return   R = {r_1, ... r_n}
\end{verbatim}
\caption{Iterative algorithm to compute for the resource allocation for every selected warehouse.}

\end{algorithm}
 
At the end of the algorithm every non relief warehouse $i \in V \setminus W$ has a resource allocation  $r_i = 0$ and every selected warehouse $j \in W$ satisfies  $$\sum_{\forall j \in W } r_j = \sum_{\forall i \in V} p_i.$$

\subsection{Delivery Assignments} \label{Delivery Assignments}
Delivery assignments are given in a manner following a division that favors efficiency over established political boundaries. Using  K-means clustering, each truck is assigned a division/cluster such that the distribution centers it would visit would be as close together as possible. This also prevents overlap in delivery assignments, and maximizes truck trips.

K-means clustering, outlined in \cite{hartigan1975clustering} is an unsupervised classification technique that splits a set of points into a defined number of groups such that points in a each group are as close together as possible. It is an iterative algorithm that starts with placing a defined number of markers on the plane. Each existing point will be assigned a marker that is closest to it (in this case, we use the Euclidean distance), with points connected to the same marker labeled as a single cluster. The centroid of each cluster is then computed and markers are relocated to those locations. Points are then assigned markers that are closest to them, and the process repeats until any succeeding iteration won't result in any cluster change. \\

\begin{figure}
\centering
\includegraphics[width=0.47\textwidth]{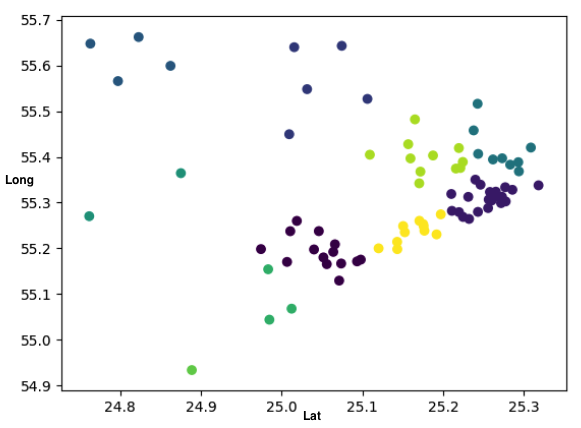}
\caption{Classification of truck delivery assignments through K-means clustering using $K=8$.  Each colored group represents a division to be traveled by a single truck. The example scatterplotis obtained by using the decimal cordinates of towns and villages of Dubai.  The x and y coordinates correspond to the Latitude and Longhitude,  respectively. }
\label{truck_assignment_using_clustering}
\end{figure}

Figure \ref{truck_assignment_using_clustering} shows the implementation of K-means clustering on the evacuation centers, with $K$ as the number of trucks available for assignment (user-defined), and the nodes present as the locations of distribution centers. Although each cluster has a varied number of assigned distribution centers, the total distance traveled by each truck remains roughly the same.

\subsection{Route Optimization} \label{route_optimization}

Optimization of disaster response logistics tasks, includes:
the optimal routing of a fleet of trucks to transport high priority humanitarian distribution over an unreliable road network, ensuring the fastest and safest routes in the  delivery of supplies.  \\

Once a cluster has been defined, an optimal route is to be instructed to them, determined by solving a Traveling Salesman Problem (TSP). TSP involves solving an NP-hard problem of an agent visiting every node in the vicinity exactly once then returning to its initial position. It aims to solve for the optimal route such that the travel time of the agent is minimized. \\

In this study, TSP will be solved using two algorithms: the two-approximation algortihm with Network X, and the TSP routing model with Google OR Tools. The performance of each algorithm will then be compared with the benchmark datasets provided by TSPLib \cite{tsplib}.

\subsubsection{Two-approximation algorithm with Network X} \label{Network X}
This algortihm uses a graphical approach to solving the TSP, and has three phases. First, a minimum spanning tree is taken from a complete, undirected graph, to ensure that every vertex is visited. Next, a pre-order DFS traversal is done to the spanning tree. Finally, with the generated route from the second phase, nodes with multiple traversals are bypassed. 

Bypassing repeated nodes in the route allows another edge to be created from the nodes connected to it. With the triangle inequality, this edge is less than or equal to the previous edges, thereby lowering the total distance. With the two-approximation algorithm, the produced path is less than twice the cost of the optimum path. The proof of the two-approximation algorithm is outlined in \cite{twoapprox} and is implemented through the NetworkX package.

\subsubsection{Routing Model with Google OR Tools} \label{Google OR Tools}
 Google OR Tools uses a routing model to solve the TSP. It is fed with an $n x n$ input matrix $M$ with $M_{ij}$ being the distance from the $i^{th}$ to the $j^{th}$ node. The algorithm is outlined in \cite{googleOR} and has several options that modify the number of times a node is visited, if a certain node can only be visited at a certain time, and if a return trip has a different time span, making this algorithm very versatile. It uses a C++ method to solve the TSP which is fed by a callback function containing the distances between nodes.

\section{Dissemination and Operation}
Instructions are disseminated to dispatched trucks through the client portal, so they can follow the pre-computed optimal route to minimize the total cost of visiting all the designated relief warehouses and dropping the necessary resources to the assigned evacuation centers.
The centralized monitoring and configuration of the relief distribution plan via a command center dashboard where:
the admin personnel can update the plan according to the total number of available delivery trucks, and
the admin personnel can update the plan with new supply distribution center locations and evacuation center locations as needed.
The estimation of storage needs and capacities for each relief distribution center based on actual demand, and
a messaging facility for real-time communication between the admin personnel and dispatchers via broadcast and feedback features.

\subsection{System Architecture}
\begin{figure}[H]
\includegraphics[width=0.47\textwidth]{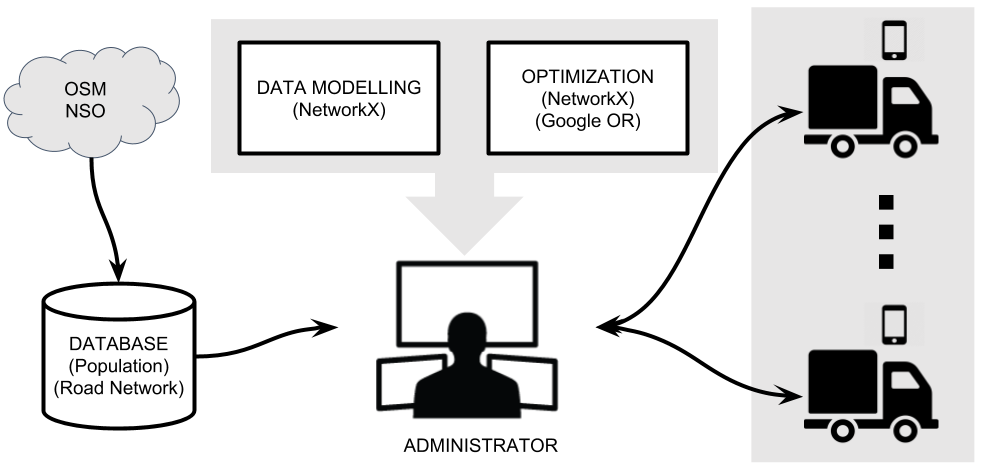}
\caption{Optima: Relief Distribution System Architecture}
\label{img:sys_arch}
\end{figure}

\section{Results and Analysis} \label{results_analysis}
%Dylan
\subsection{Route Optimization Analysis} \label{RouteOptAnalysis}
The algorithms were implemented in Python 3.6 using a 64-bit Win10 OS, Intel Core i7-8550U 1.80 GHz 8th Gen CPU with 8GB of RAM. Benchmark datasets with different node populations were taken from TSPlib, particularly-: 5 nodes, 17 nodes (gr17), 26 nodes (fri26), 48 nodes (att48), 100 nodes (kroa100), and 200 nodes (kroa200). \\

To test the performance of the algorithms outlined in \ref{route_optimization}, the cost of their proposed routes are compared to the benchmark optimal routes. The total Running time to solve each route is also monitored.

\begin{table}[]
\centering
\caption{Total running time and route cost comparison of Routing Model and Two-approximation algorithms to optimal solutions}
\label{TSPcomptable}
\begin{tabular}{cccccc}
\hline
\multirow{2}{*}{\begin{tabular}[c]{@{}c@{}}\# of \\ nodes\end{tabular}} & \multirow{2}{*}{\begin{tabular}[c]{@{}c@{}}Optimal \\ Cost\end{tabular}} & \multicolumn{2}{c}{Cost} & \multicolumn{2}{c}{Running time (s)} \\ \cline{3-6} 
 &  & \begin{tabular}[c]{@{}c@{}}Routing\\ Model\end{tabular} & \begin{tabular}[c]{@{}c@{}}Two-\\ approx\end{tabular} & \begin{tabular}[c]{@{}c@{}}Routing\\ Model\end{tabular} & \begin{tabular}[c]{@{}c@{}}Two-\\ approx\end{tabular} \\ \hline
5 & 15000 & 15000 & 30000 & 0.049 & 0.026 \\
17 & 2085 & 2085 & 2352 & 0.059 & 0.026 \\
26 & 937 & 953 & 1112 & 0.073 & 0.039 \\
48 & 33551 & 34160 & 43974 & 0.215 & 0.031 \\
100 & 21282 & 21923 & 27211.68 & 0.326 & 0.093 \\
200 & 29368 & 29188 & 38526.59 & 0.921 & 0.239 \\ \hline
\end{tabular}
\end{table}

The difference in route cost can be attributed to the algorithms reaching a local minima and passing it as a feasible route given a certain metric. This in turn saves computation time since this metric allows local minima that are close to the global minimum. The cost of the feasible routes are compared to the optimal route by the percent gap, as shown in Figure \ref{gap_comp}. \\

\begin{figure}
\centering
\includegraphics[width=0.47\textwidth]{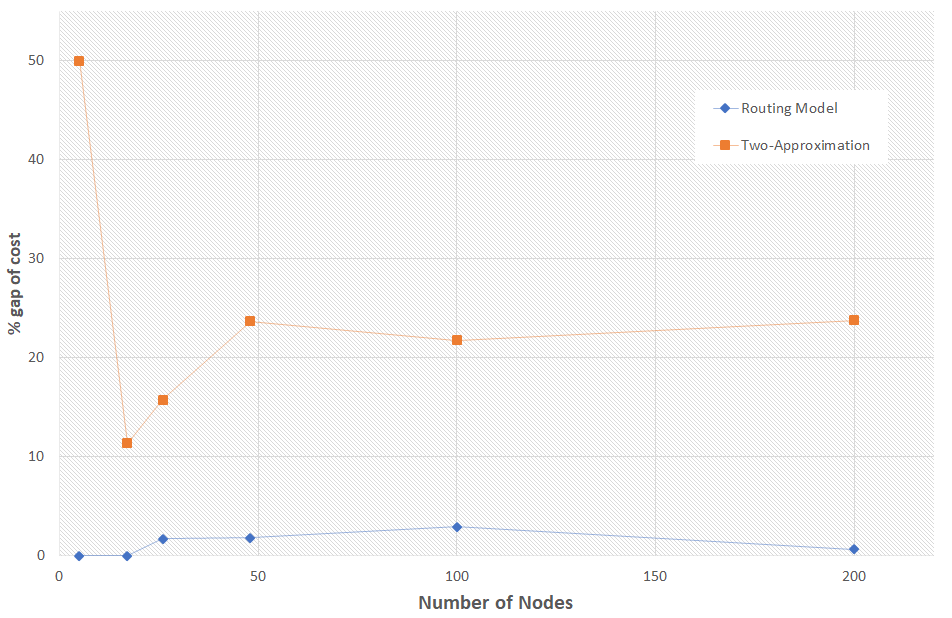}
\caption{Percent gap comparison between Routing Model and Two-approximation algorithms}
\label{gap_comp}
\end{figure}

As shown in Figure \ref{exec_time},  although the Two-approximation algorithm has a fater running time,  a higher percent gap can be observed. This is due to the constraints of this algorithm which returns a feasible route that satisfies the two-approximation condition, such as returning a route with a cost that is double that of the optimal route's cost.

\begin{figure}
\centering
\includegraphics[width=0.47\textwidth]{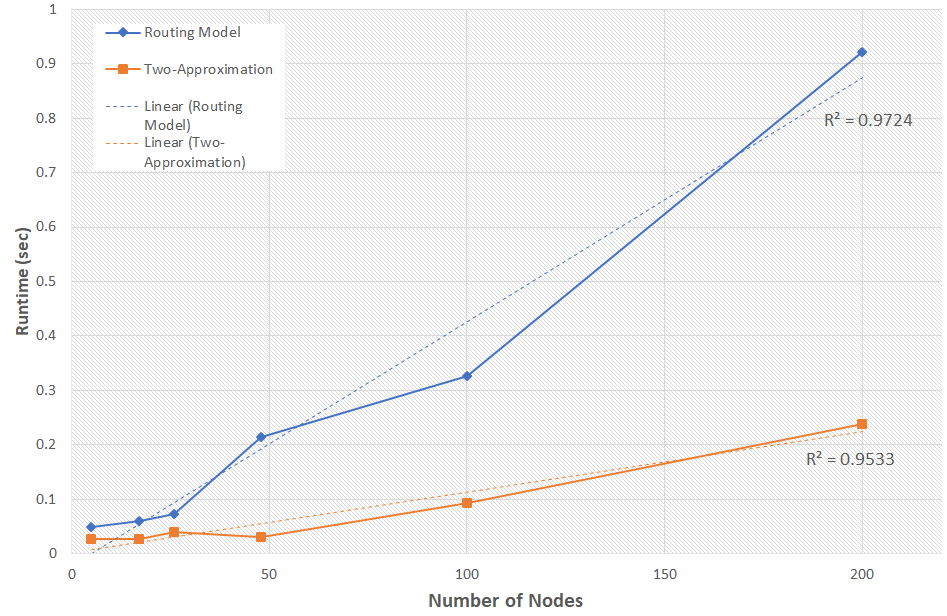}
\caption{Running time comparison between Routing Model and Two-approximation algorithms}
\label{exec_time}
\end{figure}

\section{Conclusions and Future Work} \label{conclusions}
%Jhoi

In conclusion, we provided a design and implementation of a system for creating a logistics plan for relief distribution. The system design is composed of two applications. One of which is a web application where majority of the components are deployed and the second application is a mobile application for the operators on the ground. The web application serves 4 major components of the system which produces a solution for the warehouse selection, resource allocation, delivery assignment, and route optimization. 

We introduce a simpler variant of the facility location problem called relief warehouse selection problem and provided a optimal algorithm in Algorithm 1 that produces an optimal solution.  We also provide an $O(n^2)$ solution to identify the amount of resources to drop for each of the  identified warehouses in Algorithm 2. By definition of the vertex cover,  every evacuation center is incident to at least one  warehouse. If an evacuation center is incident to $k$ warehouses, the number of  supplies for the evacuation center is equally distributed to $k$ warehouses incident to it in the computed minimum spanning tree.  

The delivery of supplies will only focus on selected warehouses and the truck assignments will be based on the natural distribution of  warehouses in the map. K-means clustering can provide warehouse assignment to $k$ available trucks. To provide a route for the truck drivers, we use the routing model implementation in Google OR Tools. 

Here, we listed down proposed additional features to the system. Since some roads may become unusable in the event of a disaster, the algorithm can be made to adapt with actual road conditions as reported by operators on the ground. The reporting mechanism will ensure safe routes through dynamic routing. 

The administrators can be given the capability to monitor the status of the relief operations as they can track the actual location of each truck in the fleet as well as the quantity of the relief goods in the warehouses. This also helps ensure transparency and accountability in the delivery of donations.

The underlying framework of the can also be used for other goods and services distribution such as water supply distribution,  laying of current and communication networks, urban planning, and others  \cite{Boonmee2017}.


\bibliographystyle{acm}
\bibliography{biblio}
\nocite{*}


\end{document}